\documentclass[%
 reprint,
 amsmath,amssymb,
 aps, prl
]{amsart}   

\usepackage{amsmath,amsthm,soul}
\usepackage{amssymb}
\usepackage{amsbsy}
\usepackage{graphics}
\usepackage{pifont}
\usepackage{epsfig}
\usepackage{setspace}
\usepackage{graphicx}
\usepackage{enumerate}
\usepackage{amsaddr}
\usepackage{epstopdf}
\usepackage{wrapfig}
\usepackage{setspace}

\usepackage{xcolor,calc}
\newcommand{\pd}[2]{\frac{\partial #1}{\partial #2}}

\newtheorem{theorem}{Theorem}

\newcommand{\bs}[1]{\boldsymbol{#1}}

\newcommand\Rey{\mbox{\text{Re}}}  
\newcommand\We{\mbox{\text{We}}}

\begin{document}
\author{Joshua Binns and Andrew Wynn}
 \address{Department of Aeronautics, Imperial College London, SW7 2AZ.}
\title{Global stability of Oldroyd-B fluids in plane Couette flow}
\email{a.wynn@imperial.ac.uk}

\date{\today}

\begin{abstract}
We prove conditions for global nonlinear stability of Oldroyd-B viscoelatic fluid flows in the Couette shear flow geometry. Global stability is inferred by analysing a new functional, called a perturbation entropy, to quantify the magnitude of the polymer perturbations from their steady-state values. The conditions for global stability extend, in a physically natural manner, classical results on global stability of Newtonian Couette flow.  
\end{abstract}

\maketitle

\section{Introduction}

Flows of polymer solutions are of wide importance to the chemical, food, cosmetic and pharmaceutical industries, being found in everyday items such as shampoo, yoghurt, and paints. Significant differences exist in the behaviour of Newtonian flows, such as air or water, and polymeric flows, dramatically exemplified by the maximum drag reduction phenomenon in which polymers fundamentally alter the nature of turbulence in pipes or channels \cite{choueiri2018}. To gain a full physical understanding of these differences requires characterizing the rich nature of viscoelastic flow instabilities, which may be purely elastic \cite{groisman2000} or arise from a combination of elastic and inertial effects \cite{samanta2013}. Crucially, elasticity can be destabilizing with instability \cite{Qin2019,pan2013} and turbulence \cite{samanta2013} observed at significantly lower Reynolds numbers $(\Rey$; the ratio of inertial to viscous forces of the solvent) than for equivalent Newtonian flows. Given this increased complexity, many fundamental questions on the stability of viscoelastic flows remain unresolved. 

Perhaps the most basic open question concerns global nonlinear stability, that is, if a viscoelastic flow is stable to initial disturbances of arbitrary amplitude. Surprisingly, even for the simplest case of Oldroyd-B fluids in the plane Couette geometry, it is unknown if there are any flow conditions under which global stability holds \cite{sanchez2022}. Numerical evidence suggests that Couette flow of an Oldroyd-B fluid is linearly stable \cite{chaudhary_2019,chokshi2009} but, unfortunately, this gives no information about global stability to arbitrary disturbances. Indeed, Newtonian Couette flow is also linearly stable, but it is well-known that finite-amplitude perturbations can trigger the transition turbulence \cite{duguet_2013} via nonlinear amplification mechanisms \cite{pringle_2010}, and this has been observed experimentally at Reynolds numbers as low as $\Rey = 325$ \cite{bottin_1998}. In the Newtonian case, however, global stability can be proven if the Reynolds number is sufficiently small, namely $\Rey  \lesssim  83$ \cite{joseph_1966,schmid2001}, which definitively rules out the possibility of persistent unsteadiness or turbulence in these conditions. 

Global stability of Newtonian Couette flow is a classical result, proven using Orr's 1907 energy method \cite{orr_1907} in which the kinetic energy of velocity perturbations is shown to be monotonically decreasing. The challenge of extending this analysis to Oldroyd-B fluids is to find a suitable elastic ``energy'' budget to quantify the magnitude of polymeric perturbations, and their potentially complex viscoelastic interactions with the velocity disturbances.  For Couette flow, it is known that natural choices such as potential elastic energy \cite{doering2006} or relative entropy \cite{jourdain2010} cannot be used to prove global stability for Oldroyd-B fluids. 

In this paper we propose a new way to quantify polymer disturbances, called a {\em perturbation entropy}, which generalizes the idea of a relative entropy functional \cite{jourdain2010} and gives the extra flexibility to prove, for the first time, global stability for Oldroyd-B fluids in plane Couette flow.  Global stability in this setting depends on three non-dimensional parameters: the Reynolds number $\Rey>0$; the Weissenberg number $\We>0$, characterizing the ratio of elastic to viscous stresses; and the ratio $0 < \beta < 1$ of the solvent viscosity to the total viscosity. We show that for any $0 < \We < 1$, nonlinear global stability holds if 
\begin{equation} \label{eq:result}
c_1 \left( \frac{1-\beta}{1-\We} \right) \We  + c_2 \Rey < \beta,
\end{equation}
where $c_1,c_2>0$ are universal constants. This is a natural extension the energy stability condition from the Newtonian case: for any Reynolds number at which Newtonian Couette flow is energy stable ($\Rey < \Rey_{\text{E}} \approx 83$), there exists a finite parameter range in $(\We,\beta)$ for which the viscoelastic Oldroyd-B model is also globally stable. 

To define the problem, suppose that a viscoelastic fluid with velocity $\bs{u} = u_1 \bs{e}_1 + u_2 \bs{e}_2 + u_3 \bs{e}_3$ is confined between two parallel plates at $x_2=0$ and $x_2=\ell$, with the bottom plate at rest and the upper plate moving in the $\bs{e}_1$ direction with velocity $U$. Periodic boundary conditions are assumed in the streamwise $\bs{e}_1$ and cross-stream $\bs{e}_3$ directions. Using $\ell$ as a length scale and $\ell/U$ as a time scale, the non-dimensional Oldroyd-B model is
\begin{equation} \label{eq:OB_pde}
\begin{split}
\frac{\partial \bs{u}}{\partial t} + (\bs{u} \cdot \nabla) \bs{u} + \nabla p  &= \frac{\beta}{\Rey} \Delta \bs{u} + \frac{1-\beta}{\Rey \We} \nabla \cdot \bs{c},\\
\nabla \cdot \bs{u} &= 0,\\
\frac{\partial \bs{c}}{\partial t} + (\bs{u}\cdot \nabla)\bs{c} = \bs{c} \cdot (\nabla & \bs{u}) + (\nabla \bs{u})^\top \!\! \cdot \! \bs{c} + \frac{1}{\We} (I-\bs{c}),
\end{split}
\end{equation}
where $p$ is the pressure and $\bs{c}(\bs{x},t) = (c_{ij}(\bs{x},t)\bs{e}_i\bs{e}_j)_{i,j=1}^3$ is a covariance tensor of the polymer orientation. The Reynolds number is $\Rey = U\ell/\eta$, where $\eta=\eta_s + \eta_p$ is the sum of the solvent and polymer viscosities,  $\beta = \eta_s / (\eta_s + \eta_p)$, and the Weissenberg number is $\We = U\lambda/\ell$, where $\lambda$ is a characteristic relaxation timescale of the dissolved polymers. 

For plane Couette flow as described, the Oldroyd-B model  \eqref{eq:OB_pde} has a steady solution 
\begin{equation} \label{eq:steady}
\bar{\bs{u}}(\bs{x}) = x_2 \bs{e}_1, \qquad \bar{\bs{c}} = \begin{pmatrix} 1 + 2\We^2 & \We & 0\\ \We & 1 & 0 \\ 0 & 0 & 1 \end{pmatrix}. 
\end{equation}
Letting $\bs{u} = \bar{\bs{u}} + \bs{v}, \bs{c} = \bar{\bs{c}} + \bs{d}$, the polymer perturbations satisfy 
\begin{equation} \label{eq:tensor_pert}
\pd{\bs{d}}{t}  + (\bs{u} \cdot \nabla )\bs{d} = s(\bs{d},\nabla \bs{u}) + s(\bar{\bs{c}},\nabla \bs{v}) - \frac{1}{\We} \bs{d},
\end{equation}
where $s(\bs{a},\bs{b}):=\bs{a} \cdot \bs{b} + \bs{b}^\top \cdot \bs{a}^\top$ is the symmetric part of the tensor product $\bs{a} \cdot \bs{b}$. The covariance matrix $\bs{c}$ is symmetric and positive definite, written $\bs{c} \succeq 0$. If $\bs{c}$ is initially symmetric and positive definite, which is assumed, then this property is preserved along solutions to \eqref{eq:OB_pde}. It follows that   $\bs{d}$ is symmetric at all times, but it is not true in general that $\bs{d} \succeq 0$. 

\section{Energy and entropy estimates}

To study global stability of the Oldroyd-B model, first consider the fluid perturbation energy $E(t):=\frac12 \|\bs{v}\|_2^2 := \frac12 \int |\bs{v}(x,t)|^2 d\bs{x}$. Differentiating, integrating by parts and using Dirichlet boundary conditions of the perturbation $\bs{v}$ gives
\begin{align} 
\dot{E}(t)  &= -\frac{\beta}{\Rey} \left\| \nabla \bs{v} \right\|_2^2 - \int v_1 v_2 d\bs{x}  - \frac{1-\beta}{\Rey \We} \int \langle \bs{d},  \nabla \bs{v} \rangle d\bs{x}, \label{eq:pert_energy}
\end{align}
where $\langle \bs{a},\bs{b} \rangle := \text{tr}(\bs{a}^\top \cdot \bs{b}) = a_{ij}b_{ij}$ is the Frobinus inner product with norm $\|\bs{a}\|_F^2 := \langle \bs{a}, \bs{a}\rangle$. 

In the Newtonian case $\beta=1$, nonlinear energy stability at sufficiently small $\Rey$ follows from \eqref{eq:pert_energy} and the Poincar\'{e} inequality. Due to the final coupling term in \eqref{eq:pert_energy},  however, extending this result to the viscoelastic case $\beta>0$ requires an appropriate quantification of the polymer perturbations. To do this, for $\alpha >0$, consider
\[
h(\bs{d}) := \frac12 \int \left[ \text{tr}(I+\alpha \bs{d}) - 3 - \log{(\det{(I+\alpha \bs{d})})} \right] d\bs{x}, 
\]
which we call a {\em perturbation entropy}. This can be interpreted as a parametrized {\em relative entropy}, which will be discussed subsequently. 

Throughout this paper, we assume that $\We \in (0,1)$ and that $\alpha =(1-\We )(1 - \We(1+\We^{2})^{-\frac12})$. It is shown in the Appendix that for these choices, $I+\alpha \bs{d} \succ 0$ for any solution to \eqref{eq:OB_pde}. Consequently, the functional $H(t):=h(\bs{d}(\cdot,t))$ is well-defined along trajectories of the system. The properties $h(\bs{0})=0$ and 
\begin{equation} \label{eq:H_bounds}
0 < h(\bs{d}) \leq \frac12 \int \| (I+\alpha \bs{d} )^{-\frac12} \alpha \bs{d}\|_F^2 d\bs{x}, \quad \forall \bs{d} \neq 0, 
\end{equation}
also proven in the Appendix, will be required. 

Given $E(t)$ and $H(t)$, our strategy is to find conditions under which the energy--entropy functional 
\[
V(t):=\alpha (1-\beta)^{-1} \We \Rey E(t)  + H(t)
\]
 decreases along trajectories of \eqref{eq:OB_pde}. To differentiate $H$, first compute
\begin{align}
\frac12 \frac{d}{dt} &\int \text{tr}(I+\alpha \bs{d}) d\bs{x} = \alpha \int d_{12} d\bs{x} + \int \left[- \frac{\alpha}{2\We} \text{tr}(\bs{d}) + \alpha \langle \bs{d}, \nabla \bs{v} \rangle  \right] \bs{dx}  \label{eq:trace_deriv}
\end{align}
which follows after integration by parts, and using the Dirichlet boundary conditions of $\bs{v}$ and $\nabla \cdot \bs{v}=0$. To differentiate the logarithmic term in $h$, we use the identity $\frac{\partial}{\partial t} \log{\det{\bs{a}}} = \text{tr}\left(\bs{a}^{-1} \frac{\partial \bs{a}}{\partial t} \right)$, which is a consequence of the Jacobi formula, twice and  \eqref{eq:tensor_pert} to get
\begin{align}
 \frac{\partial}{\partial t}  \log{(\det(I+\alpha \bs{d}))} & =  - \left( \bs{u} \cdot \nabla \right)\left( \log{(\det{(I+\alpha \bs{d})})} \right) \nonumber \\ 
& \quad +\alpha \text{tr}\left( (I+\alpha\bs{d})^{-1} \left[ s(\bs{d}, \nabla \bs{u}) + s(\bar{\bs{c}}, \nabla \bs{v}) \right] \right) \nonumber \\
& \quad - \frac{\alpha}{\We} \text{tr}\left( (I+\alpha \bs{d})^{-1} \bs{d} \right). \label{eq:log_det_deriv}
\end{align}
Combining \eqref{eq:pert_energy}, \eqref{eq:trace_deriv} and \eqref{eq:log_det_deriv}, and using $\nabla \cdot \bs{u}=0$, it follows that along trajectories of  \eqref{eq:OB_pde},
\begin{align}
\dot V(t) &= \frac{\alpha \We}{1-\beta} \left( -\beta\|\nabla \bs{v} \|_2^2 - \Rey \int v_1v_2 d\bs{x} \right) \nonumber \\
&\quad - \frac{\alpha^2}{2\We} \int \|(I+\alpha \bs{d})^{-\frac12} \bs{d} \|_F^2 d\bs{x} + \alpha \int  d_{12} d\bs{x} \nonumber \\
& \quad - \frac{\alpha}{2} \! \int \! \text{tr}\left( (I+\alpha\bs{d})^{-1} \! \left[ s(\bs{d}, \nabla \bs{u})\! + \! s(\bar{\bs{c}}, \nabla \bs{v}) \right] \right) \! d\bs{x}. \label{eq:lyap_deriv}
\end{align}

The Poincar\'{e} inequality and \eqref{eq:H_bounds} imply that, for sufficiently small $\Rey>0$, the first two terms of the above equation can be bounded above by $-cV(t)$, for some $c>0$. Hence, to infer global stability requires estimates on final two sign-indefinite terms. One difficulty is that the linear term $d_{12}$ cannot be directly bounded by $\|(I+\alpha \bs{d})^{-\frac12} \bs{d}\|_F^2$, which is quadratic as $\bs{d} \rightarrow \bs{0}$. We now explain how to avoid this obstacle. 

By the cyclic property of the trace operator, for any symmetric $\bs{a}$ it follows that $\text{tr}\left(\bs{a}^{-1} s(\bs{a},\nabla \bs{u}) \right) = \text{tr}\left(\nabla \bs{u} + (\nabla \bs{u})^\top  \right) = 2\nabla \cdot \bs{u} =0$. Applying this identity with $\bs{a}=I+\alpha \bs{d}$ gives 
\begin{align}
\text{tr}&\left( (I+\alpha\bs{d})^{-1} \left[ s(\alpha \bs{d},\nabla \bs{u}) + s(\alpha \bar{\bs{c}},\nabla \bs{v}) \right] \right) \nonumber \\ 
&= \text{tr}\big( (I+\alpha\bs{d})^{-1} [  s(I+\alpha\bs{d},\nabla \bs{u}) -s(I,\nabla \bs{u}) +  s(\alpha \bar{\bs{c}},\nabla \bs{v}) ] \big) \nonumber \\
&= \text{tr}\left( (I+\alpha\bs{d})^{-1} \left[  -s(I,\nabla \bs{u}) +  s(\alpha \bar{\bs{c}},\nabla \bs{v}) \right] \right) \nonumber \\
&= \text{tr}\left( (I+\alpha\bs{d})^{-1} \left[  -s(I,\nabla \bar{\bs{u}}) -  s(I-\alpha \bar{\bs{c}},\nabla \bs{v}) \right] \right) \label{eq:d12} 
\end{align}
and it follows that the final two, sign indefinite, terms in \eqref{eq:lyap_deriv} can be written as 
\begin{align*}
& \underbrace{ \int \alpha d_{12}  + \frac12 \text{tr}\left( (I+\alpha\bs{d})^{-1} s(I,\nabla \bar{\bs{u}}) \right) d\bs{x} }_{:=Q_{1}(\bs{d})} \\
& \qquad + \underbrace{ \frac12 \int \text{tr}\left( (I+\alpha\bs{d})^{-1} s(I-\alpha \bar{\bs{c}},\nabla \bs{v}) \right) d\bs{x} }_{:=Q_2(\bs{d},\bs{v})}.
\end{align*}
{\em A bound on $Q_{1}$}. Let $\bs{b}:=\nabla \bar{\bs{u}} + (\nabla \bar{\bs{u}})^\top$. Then 
\[
Q_1(\bs{d}) = \frac12 \int \text{tr}\left( \left( \alpha \bs{d} + (I+\alpha \bs{d})^{-1} \right) \bs{b} \right) d\bs{x}.
\]
Using that $\text{tr}(\bs{b})=0$ and the identity $A-I + (I+A)^{-1} = (I+A)^{-1}A^2$ with $A=\alpha \bs{d}$, it follows that
\begin{align*}
Q_1(\bs{d}) &= \frac12 \int \text{tr} \left( (I+\alpha \bs{d})^{-1} (\alpha \bs{d})^2 \bs{b} \right) d\bs{x}\\
						&= \frac12 \int \text{tr} \left( (I+\alpha \bs{d})^{-1} (\alpha \bs{d})^2 (\bs{b}-I) \right) d\bs{x}\\
						& \quad + \frac{\alpha^2}{2} \int \text{tr} \left( (I+\alpha \bs{d})^{-1} \bs{d}^2 \right) d\bs{x}\\
						&= \frac{\alpha^2}{2} \int \text{tr} \left( (I+\alpha \bs{d})^{-\frac12} \bs{d} (\bs{b}-I) \bs{d} (I+\alpha \bs{d})^{-\frac12} \right) d\bs{x}\\
						& \quad + \frac{\alpha^2}{2} \int \| (I+\alpha\bs{d})^{-\frac12} \bs{d} \|_F^2 d\bs{x},
\end{align*}
where in the final line, we have used that $(I+\alpha \bs{d}) \succ 0$, the cyclic property of trace, and the identity $(I+\alpha \bs{d})^{-1}\bs{d} = \bs{d}(I+\alpha \bs{d})^{-1}$. Since the eigenvalues of $\bs{b}-I$ are $\{0,-2\}$, this matrix is negative definite and it follows that $\bs{q}^\top \bs{b} \bs{q} \preceq 0$ for $\bs{q}=\bs{d}(I+\alpha \bs{d})^{-\frac12}$. Hence, $\text{tr}(\bs{q}^\top \bs{b} \bs{q}) \leq 0$ and
\begin{equation} \label{eq:K_1_est}
Q_1(\bs{d}) \leq  \frac{\alpha^2}{2} \int \| (I+\alpha\bs{d})^{-\frac12} \bs{d} \|_F^2 d\bs{x}.
\end{equation}

{\em An upper bound on $Q_2$}. Note that $\int \text{tr}(s(\bs{a},\nabla \bs{v})) d\bs{x}=0$ for any $\bs{a}$ which is independent of $\bs{x}$, since this expression is linear in $\nabla \bs{v}$, and $\bs{v}$ satisfies periodic and Dirichlet boundary conditions. Consequently, after using the identity $(I+\alpha \bs{d})^{-1} = I- (I+\alpha \bs{d})^{-1} \alpha \bs{d}$, it follows that 
\begin{align*}
Q_2(\bs{d},\bs{v}) &= -\frac12 \int \text{tr}\left( (I+\alpha \bs{d})^{-1} \alpha \bs{d} \cdot s(I-\alpha \bar{\bs{c}},\nabla \bs{v}) \right) d\bs{x}\\
& = - \int \text{tr}\left(\alpha \bs{d} (I+\alpha \bs{d})^{-1} (I-\alpha \bar{\bs{c}}) \nabla \bs{v} \right) d\bs{x}. 
\end{align*}
Using Young's inequality and the fact that $\text{tr}(\bs{a}\bs{b}) \leq \|\bs{a}\|_F\|\bs{b}\|_F$, for any $\epsilon >0$,
\begin{align*}
\left| Q_2(\bs{d},\bs{v}) \right| &\leq \frac{\epsilon \alpha^2}{2\We} \int \|(I+\alpha \bs{d})^{-\frac12} \bs{d}\|_F^2 d\bs{x} \\
& \quad +  \frac{\We}{2\epsilon} \int \|(I+\alpha \bs{d})^{-\frac12} (I-\alpha \bar{\bs{c}})\|_F^2 \| \nabla \bs{v} \|_F^2 d\bs{x}. 
\end{align*}
Next, since $I+\alpha \bs{d} \succeq I-\alpha \bar{\bs{c}}$, 
\begin{align*}
\|(I+\alpha \bs{d})^{-\frac12} (I-\alpha \bar{\bs{c}})\|_F^2 &= \text{tr}\left((I-\alpha \bar{\bs{c}})(I-\alpha \bs{d})^{-1} (I-\alpha \bar{\bs{c}}) \right) \\
 &\leq \text{tr}(I-\alpha \bar{\bs{c}}). 
\end{align*}
which gives the upper bound
\begin{align} 
|Q_2(\bs{d},\bs{v})| &\leq \frac{\epsilon \alpha^2}{2\We} \int \|(I+\alpha \bs{d})^{-\frac12} \bs{d}\|_F^2 d\bs{x} \nonumber \\
& \qquad +  \frac{\We }{2\epsilon} \text{tr}(I-\alpha \bar{\bs{c}}) \| \nabla \bs{v} \|_2^2. \label{eq:K_2_est}
\end{align}
Returning to analysis of $V(t)$, substitute the bounds \eqref{eq:K_1_est}, \eqref{eq:K_2_est} into \eqref{eq:lyap_deriv} to obtain 
\begin{align} 
\dot{V}(t) &\leq   -  \frac{\alpha^2}{2} \left[ \left( \frac{1-\epsilon}{\We} \right) -1 \right] \int \|(I+\alpha \bs{d})^{-\frac12} \bs{d} \|_F^2 d\bs{x} \nonumber \\
& \quad - \We \left[\frac{\alpha \beta}{1-\beta} \!-\! \frac{ \text{tr}(I-\alpha \bar{\bs{c}})}{2\epsilon} \! -\! \frac{\alpha C_P\Rey }{1-\beta}   \right] \|\nabla \bs{v}\|_2^2, \label{eq:lyap_deriv_2}
\end{align}
where $C_P>0$ a constant (which exists by the Poincar\'{e} inequality) such that $\left|\int v_1 v_2 d\bs{x}\right| \leq C_P \|\nabla \bs{v}\|_2^2$. We are now able to state the main result on nonlinear stability of the Oldroyd-B model to arbitrary initial perturbations. 

\begin{theorem} \label{thm:OB_global_stability}
Let $0 < \We < 1$. There exist absolute constants $c_1,c_2>0$ such that if
\begin{equation} \label{eq:gs_conditions}
\quad c_1 \, \left(\frac{1-\beta}{1-\We}\right) \We +  c_2 \, \Rey < \beta,
\end{equation}
then any solution to the Oldroyd-B system \eqref{eq:OB_pde} for plane Couette flow satisfies $(\bs{u},\bs{c}) \rightarrow (\bar{\bs{u}},\bar{\bs{c}})$ as $t \rightarrow \infty$. 
\end{theorem}
\begin{proof}
Let $\epsilon = \frac12 (1-\We)>0$. Then, by \eqref{eq:H_bounds} the first term in \eqref{eq:lyap_deriv_2} is bounded above by $-cH(t)$ for some $c>0$. Next, let $c_2:=C_P$ and let
\[
c_1:= \sup\left\{ \frac{\text{tr}(I-\alpha \bar{\bs{c}})}{ \alpha \We  } : 0 < \We < 1\right\} < \infty,
\]
noting that it can be verified numerically that $c_1 \approx 15.4$.  Then, if $(\We,\Rey,\beta)$ satisfy \eqref{eq:gs_conditions} it follows from \eqref{eq:lyap_deriv_2} and the Poincar\'{e} inequality that $\dot{V}(t) \leq -cV(t)$ for some $c>0$. Gronwall's lemma then gives $V(t)\rightarrow 0$ as $t \rightarrow \infty$ and, by \eqref{eq:H_bounds}, the system is globally stable. 
\end{proof}

\section{Discussion}

Theorem \ref{thm:OB_global_stability} naturally extends the condition for nonlinear energy stability from Newtonian to viscoelastic flows. It is not difficult to see that the optimal (smallest) constant $c_2$ in \eqref{eq:gs_conditions} is $c_2   = \Rey_{\text{E}}^{-1}$ where $\Rey_{E}$ is largest Reynolds number for which the Newtonian flow is energy stable. Consequently, Theorem \ref{thm:OB_global_stability} then implies that whenever the Reynolds number is such that the Newtonian flow is energy stable, there exist parameters $\We,\beta \in (0,1)$ for which the viscoelastic Oldroyd-B model is also nonlinearly globally stable. 

Within the current proof, the largest possible range of parameters for global stability can be visualized by numerically optimizing the free parameters $\alpha$ and $\epsilon$. The associated stability boundaries for selected values of $\beta$ are shown in Figure \ref{fig:stability_boundary}. Interestingly, for sufficiently small $\beta$ global nonlinear viscoelastic stability is possible for $\We$ and $\Rey / \Rey_{\text{E}}$ arbitrarily close to $1$ implying that, in the sense of nonlinear stability, the Oldroyd-B model smoothly transitions to the best-known Newtonian results as $\beta \rightarrow 1$. Conversely, the range of parameters for which nonlinear stability holds vanishes as $\beta \rightarrow 0$, corresponding to the limiting case of the Upper Convected Maxwell model for which $\eta_s \rightarrow 0$. 

\begin{figure}[h!]
                \includegraphics[width=.75\textwidth]{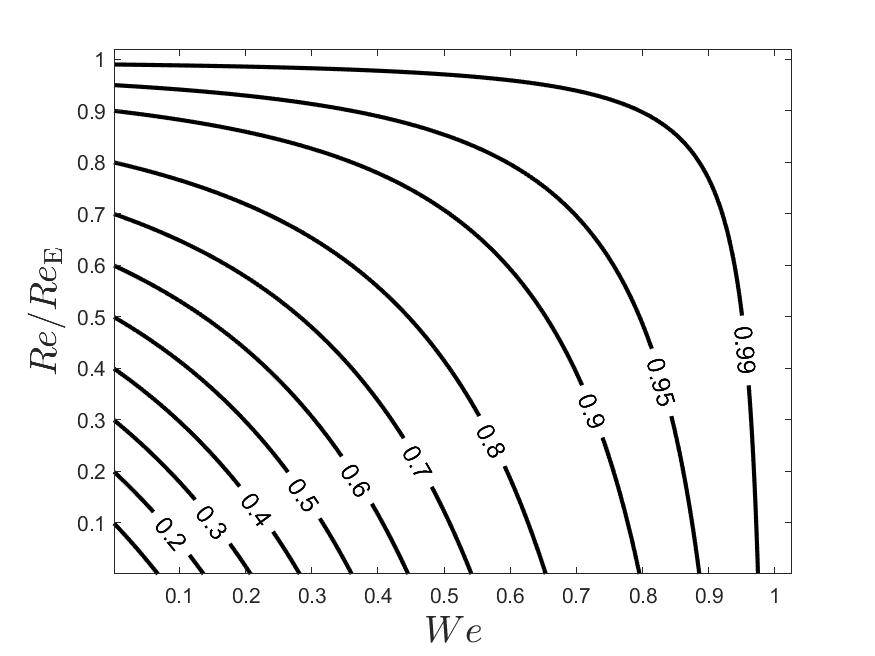}
                \caption{Stability boundaries for the Oldroyd-B model in plane Couette flow for the values of $\beta$ indicated on labelled contours. For a given point $(\rho,\omega)$ on each curve, the flow is nonlinearly stable whenever $\Rey < \rho\Rey_E$ and $\We < \omega$. \label{fig:stability_boundary}}
                
\end{figure}

We now discuss the key construction required to prove global stability of viscoelastic Couette flow, namely the perturbation entropy $h(\bs{d})$. To explain the link with entropy, we must view the Oldroyd-B model in a probabilistic context in which $\bs{c}$ is a covariance matrix $\bs{c}(\bs{x},t) = \int_{\mathbb{R}^3} \bs{X}\bs{X}^\top \psi(\bs{x},t,\bs{X}) d\bs{X}$, where $\psi(\bs{x},t,\bs{X})$ is the pdf of the local end-to-end orientation $\bs{X}(\bs{x},t) \in \mathbb{R}^3$ of the polymers at a point $\bs{x}$ in the flow domain. Letting $\overline{\psi}(\bs{X})$ be the pdf associated with the steady solution \eqref{eq:steady}, so that $\bar{\bs{c}} = \int_{\mathbb{R}^3} \overline{\psi}(\bs{X}) \bs{X}\bs{X}^\top d\bs{X}$, it is natural to quantify the deviation of $\psi$ from the steady distribution $\overline{\psi}$ using a {\em relative entropy},  
\[
\mathcal{I}\left(\psi \, | \, \overline{\psi} \right) := \int_{\mathbb{R}^3} \psi(\bs{X}) \log{\left( \frac{\psi(\bs{X})}{\overline{\psi}(\bs{X})} \right)} d\bs{X}. 
\]
If the underlying distributions are Gaussian, with $\psi \sim N(\bs{0},\bs{c})$ and $\overline{\psi} \sim N(\bs{0},\bar{\bs{c}})$, then 
\begin{align} 
& \mathcal{I}\left(\psi \, | \, \overline{\psi} \right) \nonumber \\
& = \frac12 \left( \text{tr}(I+\bar{\bs{c}}^{-1}\bs{d}) - 3 - \log{(\det{(I+\bar{\bs{c}}^{-1}\bs{d})})} \right),\label{eq:rel_entropy}
\end{align}
which is equivalent to the perturbation entropy $h$ if $\alpha$ is replaced by the $\bar{\bs{c}}^{-1}$. 

The relative entropy \eqref{eq:rel_entropy} was used in \cite{jourdain2010} to study stability of Oldroyd-B fluids. However, this approach failed (see \cite[Remark 10]{jourdain2010}) for non-Dirichlet boundary conditions, including the case of plane Couette flow considered here. To view the perturbation entropy $h(\bs{d})$ in an equivalent manner, note that if $\bs{X}_\alpha \in \mathbb{R}^3 \sim N(\bs{0},I-\alpha \bar{\bs{c}})$ is a random Gaussian vector with covariance matrix $I-\alpha \bar{\bs{c}}$, then $\bs{X}_\alpha + \sqrt{\alpha}\bs{X} \sim N(\bs{0},I+\alpha \bs{d})$. Consequently, 
\[
h(\bs{d}) = \mathcal{I}\left( \bs{X}_\alpha + \sqrt{\alpha} \bs{X} \, | \, N(\bs{0},\bs{I}) \right)
\]
and the perturbation entropy can be interpreted as relative entropy of $\bs{X}_\alpha + \sqrt{\alpha}\bs{X}$ with respect to a standard Gaussian distribution. It appears that this structure of perturbation entropy, with the freedom to choose the parameter $\alpha$, is more naturally suited to studying global stability of Oldroyd-B fluids.

\section{Conclusions}

In summary, we have proven sufficient conditions for global stability of Oldroyd-B fluids in the plane Couette geometry. The range of non-dimensional parameter values, given by \eqref{eq:result}, for which global stability holds includes a continuous parameter range interpolating between Upper Convected Maxwell and Newtonian flows. Our results are consistent with the limited existing numerical and experimental observations of viscoelastic channel flows, which indicate instability for either sufficiently high Reynolds number or Weissenberg number (e.g., at $\Rey=2000,\We=1$ in \cite{chokshi2009}; or at $\Rey = 0.01, \We \gtrsim 5$ in \cite{Qin2019}). 

That a gap exists between provable conditions for global stability and experimental observations is not surprising, since this is also true in the simpler case of Newtonian Couette flow. Closing this parametric gap is challenging. It is only very recently \cite{fuentes_2022} that generalizations to Orr's 1907 energy method have enabled quantitative improvements to the range of Reynolds numbers for which even 2D Newtonian Couette flow is provably globally stable. One open question is whether this approach can be coupled with the new class of perturbation entropy functionals introduced in this paper, in order to widen the parametric range of provable global stability of viscoelastic Couette flow. Further extensions to global stability analysis of other canonical geometries such as Poiseuille  and pipe flow are also of interest.  

Beyond global stability, it is an open question whether perturbation entropy functions can be used to study transition to turbulence in viscoelastic flows. For example, they may be used as a metric to quantify worst-case disturbances, as opposed to the energy-based methods typically employed \cite{lieu_2013}. Furthermore, our approach may open the door to rigorous proofs of scaling laws for turbulent statistics of viscoelastic flows, such as the relation $f_\eta \sim \We^\frac13$ between friction factor $f_\eta$ and Weissenberg number observed for microchannel flow in \cite{Qin2019}, by embedding perturbation entropy functionals into the background method \cite{fantuzzi_2022} formalism.

\section{ \label{sec:app_A} Appendix}

{\em A lower bound on $\lambda_{\text{min}}(I+\alpha \bs{d})$.} The eigenvalues of $I-\alpha \bar{\bs{c}}$ are 
\[
\left\{1-\alpha,1-\alpha\left[1+\We^2 \pm \We\sqrt{1+\We^2}\right] \right\}.
\]
 Hence, if $\alpha:=(1-\We)(1-\We/\sqrt{1+\We^{2}})$ then 
\begin{align*}
\lambda_{\text{min}}(I+\alpha\bs{d}) &= \lambda_{\text{min}}(I-\alpha \bar{\bs{c}} + \alpha \bs{c}) \\
&\geq \lambda_{\text{min}}(I-\alpha \bar{\bs{c}})=\We,
\end{align*}
where we have used that $\bs{c} \succeq 0$. Hence, $I+\alpha \bs{d} \succ 0$. $\qed$

{\em Proof of \eqref{eq:H_bounds}.} Let $\lambda_i >0$ be the eigenvalues of  $I+\alpha\bs{d} \succ 0$. Then 
\[
h(\bs{d}) = \frac12 \int \sum_{i=1}^3 \left[ \lambda_i - \log{\lambda_i} -1 \right] d\bs{x}.
\]
Now,  $x\mapsto x - \log{x}-1 \geq 0$ for any $x>0$, with equality only when $x=1$. Consequently, $h(\bs{d}) \geq 0$ with equality if and only if all eigenvalues of $I+\alpha \bs{d}$ are equal to one, which implies $I+\alpha\bs{d}=I$ and $\bs{d}=0$. For the upper bound, let $A:=\alpha \bs{d}(\bs{x},t)$. Then since $\|(I+A)^{-\frac12}A\|_F^2 = \text{tr}((I+A)^{-1}A^2)$, the upper bound holds if for every $\bs{x},t$, 
\[
\log{(\det{(I+A))}} \geq \text{tr}( A - (I+A)^{-1}A^2) = 3 - \text{tr}((I+A)^{-1})
\]
The above inequality holds if $\sum_{i=1}^3 \left[ \lambda_i^{-1} - \log{\lambda_i^{-1}} -1 \right] \geq 0$, which is true by positivity of $\lambda_i$. $\qed$

\bibliographystyle{amsplain}
\bibliography{OB_bib}

\end{document}